\documentclass[12pt]{article}
\usepackage{amsmath,amssymb,amsthm}
\usepackage{graphics,epsfig}
\usepackage{hyperref}
\usepackage{natbib}
\usepackage{color}
\usepackage{graphicx}
\usepackage{caption}
\usepackage{subcaption}
\usepackage{float}
\usepackage{mathrsfs}
\usepackage{multirow}
\usepackage{comment}
\usepackage{natbib}
\usepackage{setspace}

\def \E {\mathbb{E}}

\newtheorem{definition}{\bf Definition}

	\newtheorem{theorem}{\bf Theorem}
	\newtheorem{thm}[theorem]{\bf Theorem}


\begin{document}
	
\title{\bfseries Dissecting Multi-Level Pricing Schemes in the Context of eCW Client Engagement}

\author{
	Paramahansa Pramanik\footnote{Corresponding author, {\small\texttt{ppramanik@southalabama.edu}}}\; \footnote{Department of Mathematics and Statistics,  University of South Alabama, Mobile, AL, 36688,
		United States.}
	\and
	Joel Graff\footnote{Intelligent Medical Objects Inc.(IMO), Rosemont, IL 60018}
	\and 
	Mike Decaro\footnotemark[3]
}

\date{\today}
\maketitle

\begin{abstract}
This paper presents a usage-based pricing framework for the Intelligent Medical Objects’ ProblemIT™ Portal utilized by eClinicalWorks (eCW) clients. The approach begins by determining a stable monthly unit price per request, estimated as the median from semi-parametric Bayesian cubic smoothing spline analyses covering the period November 2015 to December 2016. Clients are subsequently segmented into eight volume-based tiers, with total charges computed by multiplying the derived median unit price by each client’s total request count. Examination of the dataset reveals that 806 accounts with a single registered user and 470 accounts with two registered users both exhibit disproportionately high request volumes. The proposed model incorporates adjustments to account for these anomalies.
\end{abstract}


\section{Introduction.}
eClinicalWorks (eCW) \citep{edsall2008user} is widely regarded as a leading provider of Electronic Health Record (EHR) systems, delivering integrated clinical and administrative solutions to a broad user base. Its technologies are adopted by more than 125,000 clinical specialists and approximately 850,000 additional healthcare professionals across diverse practice environments \citep{evans2016electronic}. The company’s platforms support a range of functions, including patient record management, clinical decision support, and revenue cycle optimization. Among these offerings, the Intelligent Medical Objects’ ProblemIT™ solution serves as a specialized portal that facilitates precise clinical terminology mapping, thereby enhancing documentation quality, billing accuracy, and interoperability between healthcare systems. This portal is central to the present study, which seeks to understand usage patterns and develop a robust pricing framework grounded in empirical data \citep{dasgupta2023frequent,hertweck2023clinicopathological,khan2024mp60}.

The analytical framework is constructed by first estimating the total volume of ProblemIT™ portal requests each month, then dividing this value by the monthly invoiced amount to derive a consistent unit price per request. November 2015 was selected as the baseline for the study, given that October 2015 marked the nationwide adoption of a new medical billing code system, an industry-wide transition that substantially altered how clients engaged with the ProblemIT™ solution \citep{kakkat2023cardiovascular,khan2023myb}. This transition, widely associated with the implementation of ICD-10 coding standards, represented a structural shift in documentation and coding behavior, potentially influencing both request frequency and billing dynamics \citep{vikramdeo2023profiling,vikramdeo2024abstract}. To ensure temporal consistency and avoid distortions from the immediate effects of this policy change, the dataset is confined to a 14-month period spanning November 2015 through December 2016. This post-transition interval provides a stable analytical window in which to observe typical client behavior, allowing the price-per-request metric to be evaluated under steady-state conditions rather than transitional anomalies \citep{jackson2020health,pramanik2022lock}.

The first stage of our analysis focuses on deriving a stable estimate of the monthly price per request. In this context, stability refers to a value that exhibits minimal variation over time, such that fluctuations, if present, are only detectable beyond the fourth decimal place \citep{pramanik2021optimala}. Establishing this level of precision is essential to ensure that subsequent pricing calculations are not influenced by short-term volatility or transient anomalies. To formally evaluate the stability of the price-per-request series, we apply two widely used statistical tests for unit roots: the augmented Dickey-Fuller (ADF) test \citep{dickey1979distribution} and the Phillips-Perron (PP) test \citep{phillips1988testing}. The results from both procedures indicate that the series is stationary in its first-order difference, confirming that the underlying process meets the stability criteria outlined above \citep{pramanik2024parametric}.

Having established stationarity, we proceed to identify the most appropriate time series model for capturing the observed dynamics. Competing models are compared using the Akaike Information Criterion (AIC) \citep{akaike2003new}, with lower values indicating superior fit. This comparison reveals that a first-order moving average (MA) model yields the lowest AIC among the candidates considered. Nevertheless, while the MA(1) model provides a satisfactory baseline, our objective also includes minimizing the mean squared error (MSE) in the fitted values. To this end, we implement a semi-parametric Bayesian cubic smoothing spline approach \citep{zhang1998semiparametric}, which demonstrates a superior fit to the observed data \citep{perou2000molecular,prat2010characterization}. To guard against the undue influence of outliers and to reduce the risk of overfitting, both of which could distort cost estimates, we adopt the median of the spline’s predicted values as the representative monthly price per request. This approach combines statistical rigor with robustness, ensuring that the resulting price measure is both stable and reliable for downstream pricing model development \citep{pramanik2021scoring}.

Once the predicted monthly price per usage has been established through the preceding statistical procedures, the next stage of the analysis involves applying this value to the observed distribution of portal activity. Specifically, the total number of requests generated by all clients over the study period is organized into eight distinct volume-based tiers, as summarized in Table 1. For each tier, the total number of requests is multiplied by the estimated monthly price per request, yielding the corresponding cost for that usage category \citep{pramanik2020optimization,pramanik2023semicooperation}. This tiered approach provides a structured framework for assessing how clients are distributed across varying levels of portal activity, thereby enabling the identification of patterns and anomalies in usage behavior. Of particular interest are instances where eCW customers maintain a relatively small number of user licenses, sometimes only one or two, yet generate a disproportionately high total volume of requests \citep{pramanik2024estimation,vikramdeo2024mitochondrial}. Such cases, when contrasted with clients whose request volume more closely aligns with their reported license count, suggest potential inefficiencies or imbalances in the existing cost allocation structure. The identification of these outliers is made possible by direct comparison between the proposed usage-based pricing model and the current pricing framework, which determines charges primarily on the basis of the reported number of end users rather than actual system utilization \citep{pramanik2024motivation}. This comparison not only highlights disparities in client usage patterns but also offers insight into how the adoption of a usage-driven tiered pricing model could address these discrepancies and promote a more equitable alignment between service consumption and cost \citep{pramanik2020motivation}.

The detection of accounts with minimal license counts but unusually high request volumes has important implications for both operational efficiency and revenue management within the eCW client base \citep{pramanik2024bayes,bulls2025assessing}. From a financial perspective, such usage patterns indicate that the current license-based pricing model may underrepresent the true level of system demand for these clients, resulting in a potential shortfall in revenue relative to actual resource utilization. From an equity standpoint, this misalignment can inadvertently lead to cross-subsidization, whereby low-usage clients effectively bear a proportion of the costs generated by high-usage clients who are charged less than their relative system consumption would justify. Incorporating a usage-based tiered structure mitigates this imbalance by ensuring that charges are proportionate to the intensity of portal activity, thereby distributing costs more equitably across the customer base \citep{pramanik2023cont,pramanik2024estimation}. Moreover, identifying these high-usage anomalies offers strategic insights for refining licensing agreements, guiding client outreach, and informing future product design to better match system capabilities with customer needs \citep{pramanik2024estimation1}. By linking cost recovery more closely to actual demand, the proposed pricing model not only addresses inefficiencies in the current framework but also supports sustainable growth, improved resource allocation, and enhanced transparency in client billing practices.

\section{Background and Rationale.}

An examination of the dataset reveals a noteworthy association between the number of active license holders and the corresponding total number of portal requests, with the latter generally exhibiting a lower bound that appears to be determined by the license count. This pattern suggests that, under typical conditions, organizations with more licensed users tend to generate higher overall request volumes, while those with fewer licenses exhibit proportionally reduced usage \citep{pramanik2023cmbp,pramanik2023optimization001}. However, several notable exceptions to this trend emerge, wherein certain customers maintain a limited or even nonexistent number of licensed users yet record exceptionally high usage levels. These cases are particularly striking given the magnitude of the discrepancy between reported licenses and observed activity. A prominent example occurred in January 2016, when Central Ohio Primary Care Physicians reported zero licensed users yet registered a total of 229,565 requests during that month alone \citep{pramanik2023path}. This phenomenon is not confined to a single organization; multiple clients, including Methodist Healthcare Physician Services, Central Ohio Primary Care Physicians, and Comprehensive Pain Specialists, consistently displayed similar anomalies, with reported license counts of zero while generating monthly request volumes exceeding 100,000 \citep{pramanik2021,pramanik2021consensus}. Such findings highlight a structural mismatch between the licensing data and actual system utilization, raising important questions about the accuracy of reported license counts, the mechanisms through which system access is managed, and the implications for both pricing models and resource allocation.

Further analysis of the dataset uncovers additional patterns that deviate from the expected relationship between the number of licensed users and total portal requests. Several customers maintain only a single reported license yet generate substantial monthly request volumes, well beyond what would typically be anticipated for such a small user base. Notable examples include IHealth Family Care, Sovereign Health Medical Group, Ospa, Todd J Kazdan DO PA, Pedro E. Estorque Jr., and Samuel A. Boliver M.D. Inc. Among these, Sovereign Health Medical Group recorded the highest activity level within this category, submitting 44,479 requests in December 2016 despite holding only one license. This level of usage suggests either intensive use by a single licensed account or potential discrepancies in the reporting of license data relative to actual access patterns \citep{pramanik2022stochastic}. By contrast, organizations with substantially larger license counts display request levels that, while high, appear more proportional to their licensed user base \citep{pramanik2024stochastic}. For instance, WellMed Medical Management, Inc., which held 450 licenses, reached its peak monthly activity in August 2016 with 294,622 requests. This example demonstrates a more conventional alignment between license volume and total requests, further emphasizing the anomalies observed in smaller-license accounts \citep{pramanik2025stubbornness}.

In addition to these one-license high-usage cases, the data reveal another set of outliers involving customers with no recorded licenses who nevertheless maintain significant portal activity. Entities such as JSA Healthcare Corporation, Children’s National Medical Center, First Choice Health Centers, Inc., Sheridan Healthcorp, and The Toledo Clinic exhibit this pattern, consistently registering monthly request counts exceeding 20,000 despite reporting zero licenses \citep{pramanik2025factors}. Among these, The Toledo Clinic represents the most extreme instance, with a request volume of 292,270 in January 2016. The persistence of such anomalies across multiple organizations indicates that the observed discrepancies are not isolated incidents but rather a recurring feature of the dataset \citep{pramanik2025optimal}. These findings raise critical questions regarding the accuracy and completeness of license reporting mechanisms, the potential for shared or unlicensed access pathways, and the implications for both pricing accuracy and system resource allocation \citep{pramanik2016,pramanik2021thesis}. Identifying and addressing such cases is essential for developing a fair and effective usage-based pricing model, as they represent points where current billing structures may fail to capture the true scale of system utilization \citep{maki2025new}.

The observed usage patterns indicate that customers with relatively few licensed users often generate disproportionately high request volumes, a phenomenon that is not adequately addressed by the existing user-based pricing framework \citep{polansky2021motif}. Under the current model, charges are determined primarily by the number of reported licenses, meaning that organizations with extensive system utilization but minimal licenses may contribute significantly to overall demand on IMO services while paying comparatively less than their actual usage would warrant. This misalignment creates an inherent inequity in cost distribution, as the financial burden is not proportionately shared among clients according to the intensity of their system usage \citep{pramanik2023cont,pramanik2024estimation}. Furthermore, it introduces inefficiencies in revenue recovery, as high-volume clients effectively underpay relative to their resource consumption, while low-volume clients may bear a relatively greater share of the cost. Addressing this imbalance requires rethinking the pricing framework so that it better reflects real-world utilization patterns rather than relying solely on static license counts \citep{bailey2017structural,menikdiwela2022association}.

In response to this disparity, we propose replacing the traditional user-based pricing model with a request-based approach that ties costs directly to system usage. The implementation of such a framework begins with the derivation of a stable, predicted monthly price per request, following the methodology outlined in \citep{pramanik2024dependence}. Once established, this unit price is applied to each usage group by multiplying it by the total number of requests recorded for that group during a given month. This calculation effectively yields a pricing structure in which the cost assigned to a customer is equal to the predicted monthly price per request multiplied by their total request volume, thereby ensuring that charges scale proportionally with actual usage \citep{pramanik2025stubbornness,pramanik2025factors}. Such a model enhances fairness by aligning payments with service consumption and has the potential to improve both transparency and efficiency in billing, while also providing a more sustainable basis for revenue generation and resource allocation.

Integrating the request-based framework with the proposed eight-tier categorization system enables a more granular and equitable allocation of costs across the client base. By segmenting customers into discrete usage tiers according to their total monthly request volumes, the model not only captures variations in demand but also allows for systematic identification of anomalies, such as clients with zero or minimal licenses who nonetheless exhibit disproportionately high portal activity \citep{pramanik2023optimization001,pramanik2025strategies}. In this structure, each tier is assigned a total cost derived from the stable predicted monthly price per request multiplied by the aggregate number of requests for that tier, ensuring that clients within the same tier are billed consistently relative to their consumption. This tiered approach serves a dual purpose: it simplifies the practical implementation of request-based billing for administrative purposes while simultaneously addressing the inequities inherent in the user-based model \citep{pramanik2025impact,pramanik2025strategic}. High-usage, low-license clients previously undercharged under the traditional system are thus brought into alignment with their true level of service utilization, fostering fairness, transparency, and improved revenue recovery.

\section{Data Examination and Interpretation.}

The primary data sources employed in this analysis are extracted from multiple tables housed within the IMO Data Warehouse (IDW), each providing distinct but complementary information necessary for constructing the proposed pricing model. The \texttt{channel\_sales\_fact} table contains key operational and financial variables, including the number of end users, the effective and expiration dates of licenses, and the invoiced amounts associated with each account. Customer-specific identifiers are obtained from the \texttt{customer\_dim} table, which provides client names along with globally unique identifiers (GUIDs) for the ProblemIT\texttrademark\ portal, ensuring accurate linkage across datasets. The distribution channel information, specifically identifying eCW clients is stored in the \texttt{channel\_group\_dim} table. Additional transactional and contractual details are drawn from the \texttt{mscrm.sales\_order\_detail\_entity} table, while request-level activity data, including the number of portal searches or queries, is obtained from the \texttt{log.org\_stats} table. An important constraint observed in the licensing data is that the license term is capped at a maximum of twelve months, necessitating monthly aggregation for consistent temporal analysis. To compute the total invoiced amount for each month, it is essential to determine the number of active licenses at that point in time \citep{pramanik2024dependence,pramanik2024measuring}. This requires identifying the number of existing users carried over from the previous month, adding the number of new licenses issued during the current month, and subtracting the number of licenses set to expire at the start of the month. The resulting figure represents the total number of licenses active in that month and serves as a foundational input for calculating monthly invoiced amounts and, subsequently, the unit price per request \citep{pramanik2024parametric,yusuf2025prognostic}.

To derive the monthly price per request, we compute the ratio of the total number of portal requests generated during a given month to the total invoiced amount recorded for that same month. This metric serves as a standardized unit price, enabling comparisons of usage costs across different periods and customer groups. To illustrate the temporal behavior of this measure, Figures~\ref{fig:l1} and~\ref{fig:l2} display two separate visualizations of the monthly price per request, each based on a different observational window. The first plot, shown in Figure~\ref{fig:l1}, spans a full twenty-four months, covering the period from January 2015 through December 2016. This extended time frame provides a broader historical perspective, capturing variations both before and after notable changes in customer behavior. In contrast, the second plot, presented in Figure~\ref{fig:l2}, focuses exclusively on the fourteen-month interval from November 2015 through December 2016. This shorter time frame was selected to isolate a post-transition period of interest, as November 2015 marks the point immediately following the implementation of a new medical billing code system in October 2015. The comparative presentation of these two intervals allows for the examination of both long-term and post-transition trends in the monthly price per request, thereby offering insight into the stability and potential shifts in cost structure over time.

\begin{figure}[htbp]
	\centering
	\includegraphics[width=7cm]{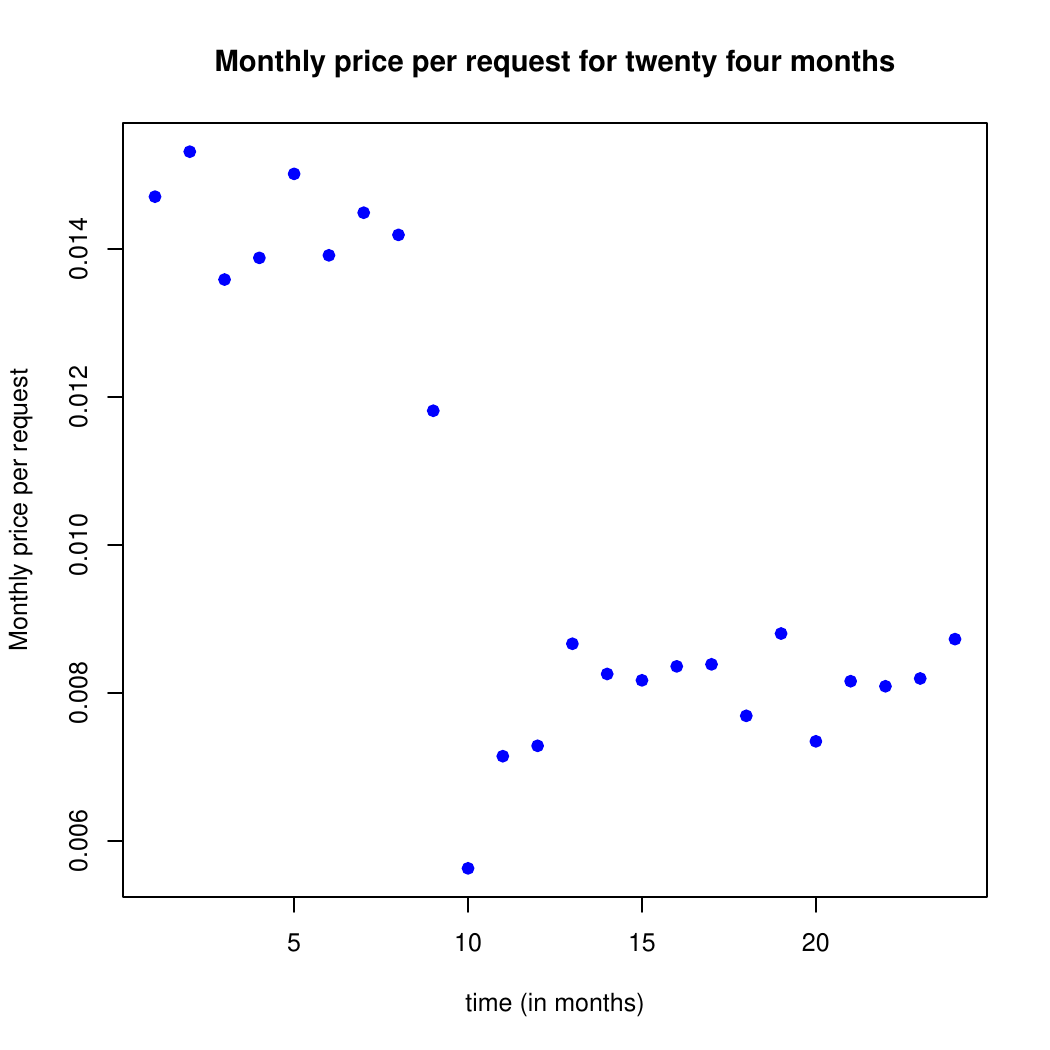}
	\caption{Twenty four month time period.}
	\label{fig:l1}
\end{figure}

\begin{figure}[htbp]
	\centering
	\includegraphics[width=7cm]{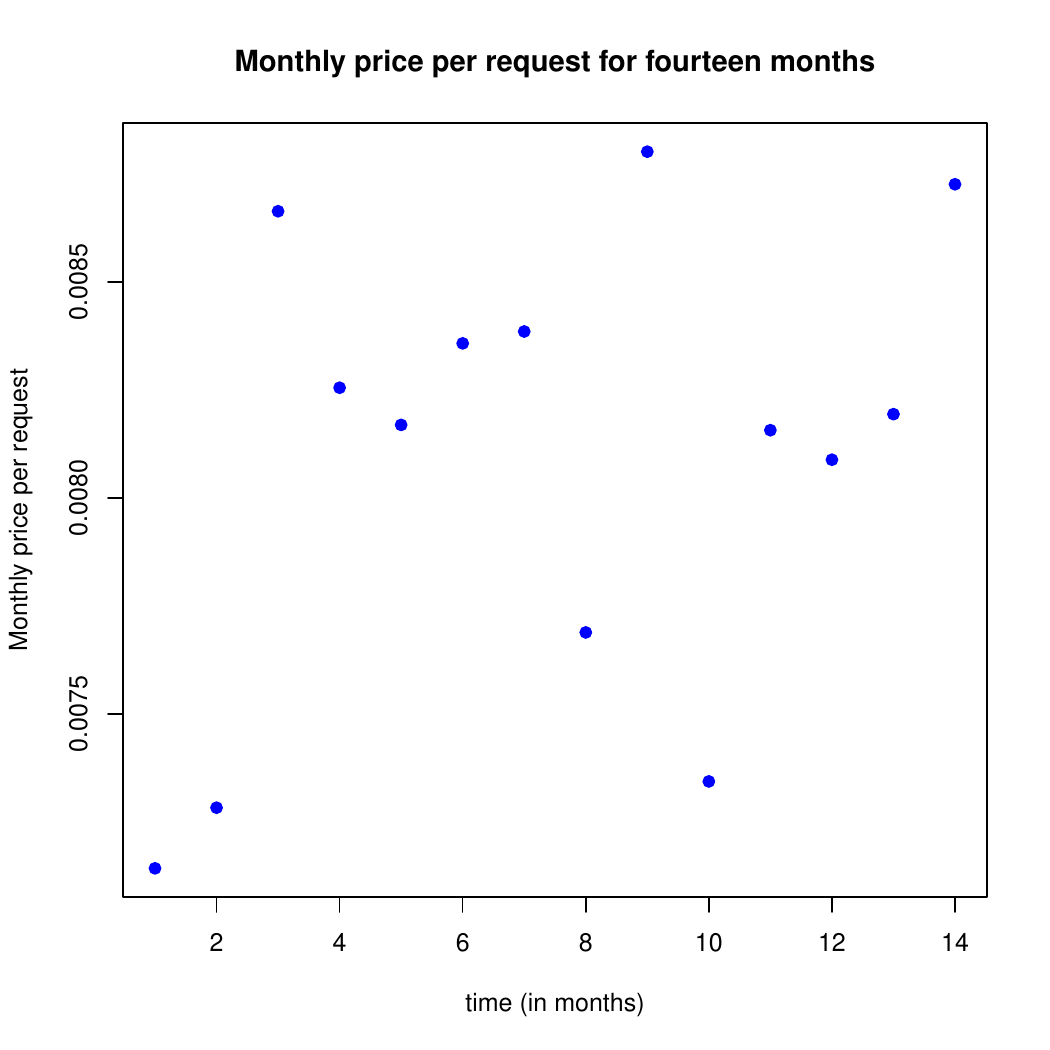}
	\caption{Twenty four month time period.}
	\label{fig:l2}
\end{figure}
The selection of the time frame illustrated in Figure~\ref{fig:l1} is driven primarily by two considerations. First, prior to 2015, the available dataset does not provide sufficient information on the number of active users, limiting the reliability of any computed price-per-request estimates for earlier periods. Second, October 2015 marks the “ICD-10 go-live” event, during which the transition to the ICD-10 medical billing code system produced a pronounced surge in overall portal usage. This spike is clearly visible in the early portion of the plot and is followed by a period of decline that transitions into a more stable pattern throughout 2016. While data from January to May 2017 are available, these months are intentionally excluded from Figure~\ref{fig:l1} because they are reserved for a separate analysis in which predicted monthly prices per request are directly compared with the corresponding observed values. An additional observation in the longer time frame is the emergence of an upward trend in the monthly price per request beginning in December 2016. This increase is likely linked to the natural fluctuations of the business cycle and is expected to reverse after reaching a certain threshold. In Figures~\ref{fig:l1} and~\ref{fig:l2}, the vertical axis represents the monthly price per request, while the horizontal axis measures time in months from the respective starting points, January 2015 for Figure~\ref{fig:l1} and November 2015 for Figure~\ref{fig:l2}. A closer inspection of Figure~\ref{fig:l1} reveals two visually distinct clusters of data points: one located in the upper left region of the plot and another in the lower right. The upper left cluster, associated with the immediate post-ICD-10 implementation surge, is treated as an outlier and removed from subsequent stability analysis to avoid distortion of the price estimates. The lower right cluster displays a more consistent distribution of values, indicating a steadier relationship between usage and pricing. As a result, the refined dataset used for Figure~\ref{fig:l2} contains only 14 monthly observations corresponding to this stable post-transition period. Within this shorter interval, the points remain somewhat dispersed, reflecting the inherent variability in monthly portal usage even after excluding the extreme fluctuations of the earlier period.

\section{Model.}

In this section, we present the methodological approach used to identify the most appropriate model for deriving a stable estimate of the monthly price per request, which forms the foundation for the proposed request-based pricing framework. The primary objective is to ensure that the calculated price per request remains consistent over time, minimizing the influence of short-term fluctuations while accurately reflecting underlying usage patterns. After evaluating several potential modeling strategies, we adopt a cubic natural smoothing spline approach, as this method demonstrates superior predictive accuracy when applied to the current dataset. Its flexibility in capturing nonlinear relationships, combined with its ability to produce smooth and stable estimates, makes it particularly well-suited for the conditions under consideration. The resulting predicted monthly prices per request, generated by this model, are then integrated into the request-based structure to determine cost allocations across different usage groups.

\subsection{Testing of Stability of Monthly Prices.}
To assess whether the trend of price per request in the scatter plot in figure \ref{fig:l2} remains stable over time, we further analyze the plots of the autocorrelation functions (shown in figure \ref{acfplot}). Using our historical data, we represent the monthly price per request at time point $t$ as $y_t$. For another time point $s$, the autocorrelation is defined as follows,
\begin{align}
\label{1}
\rho_{t,s}=\frac{\E[(y_t-\mu_t)(y_s-\mu_s)]}{\sigma_t \sigma_s},
\end{align}
where $\E$ is the expected value operator, $y_s$ monthly price per request at time s, $\mu_t$, $\mu_s$, $\sigma_t$ and $\sigma_s$ are the means and standard deviations of $y_t$ and $y_s$, respectively. Since it is a correlation, then $\rho_{t,s}\in[-1,1]$.

As the autocorrelation $\rho_{t,s}$ lies within the range $[-1,1]$, we can determine whether to use a trend-stable (autoregressive) model or a trend-stochastic (moving average) model. Stochasticity refers to the randomness of a system. In other words, trend stochasticity means that the trend of the system itself behaves like a random variable with some probability. If we assume trend stability, we focus solely on autoregressive models, while for trend stochasticity, we consider the moving average model. Additionally, there may be cases where the trend of price per request is a combination of both models. In the following sections, we will discuss these models individually.

Consider the first order autoregressive  model (AR(1))
$$y_t=\rho_{t,t-1}\ y_{t-1}+\epsilon_t,$$
where $\rho_{t,t-1}$ represents the autocorrelation between time $t$ and $(t-1)$ and $\epsilon_t$ is an error term (termed as \emph{white noise}) such that $\epsilon_t\overset{i.i.d}{\sim}\ N(0,\sigma^2)$. The randomness of $y_t$ comes from  $\epsilon_t$. Clearly, the error term follows a normal distribution with mean zero and constant variance $\sigma^2$. Since the variance is constant over time, we call it as homoskedastic variance. More generally, when the variance changes over time (i.e., $\sigma_t^2$ instead of $\sigma^2$)  is called as heteroskedastic variance. In this paper we assume homoskedastic variance. If $|\rho_{t,t-1}|=1$ then the above AR model yields, 
$$y_t=y_{t-1}+\epsilon_t.$$
Thus,  the variation of $y_t$ comes not only from the error term but also from $y_{t-1}$. Therefore, the system has a unit root problem. In other words, $y_t$ only depends on the randomness of the process. From the previous discussion we know that the randomness come from the error part  $\epsilon_t$; with $|\rho_{t,t-1}|=1$ it has more influence on $y_t$ than with $|\rho_{t,t-1}|<1$. Therefore, AR(1) is applicable when $|\rho_{t,t-1}|<1$. In other words, when $\rho_{t,t-1}=1$ autoregressive model is not appropriate to use. Therefore, we use moving average model of first degree as 
$$ y_t=\epsilon_t+\theta\ \epsilon_{t-1},$$
where $\epsilon_t$ and $\epsilon_{t,t-1}$ are white noises. It is important to note that AR and MA are two completely separate models under $|\theta|<1$ ; MA(1)=AR($\infty$), and vice versa (see Theorem 1 in the appendix).

\begin{figure}[htbp]
	\centering
	\subfloat[Autocorrelation function (ACF) of monthly price per request.]{%
		\includegraphics[width=5.5cm]{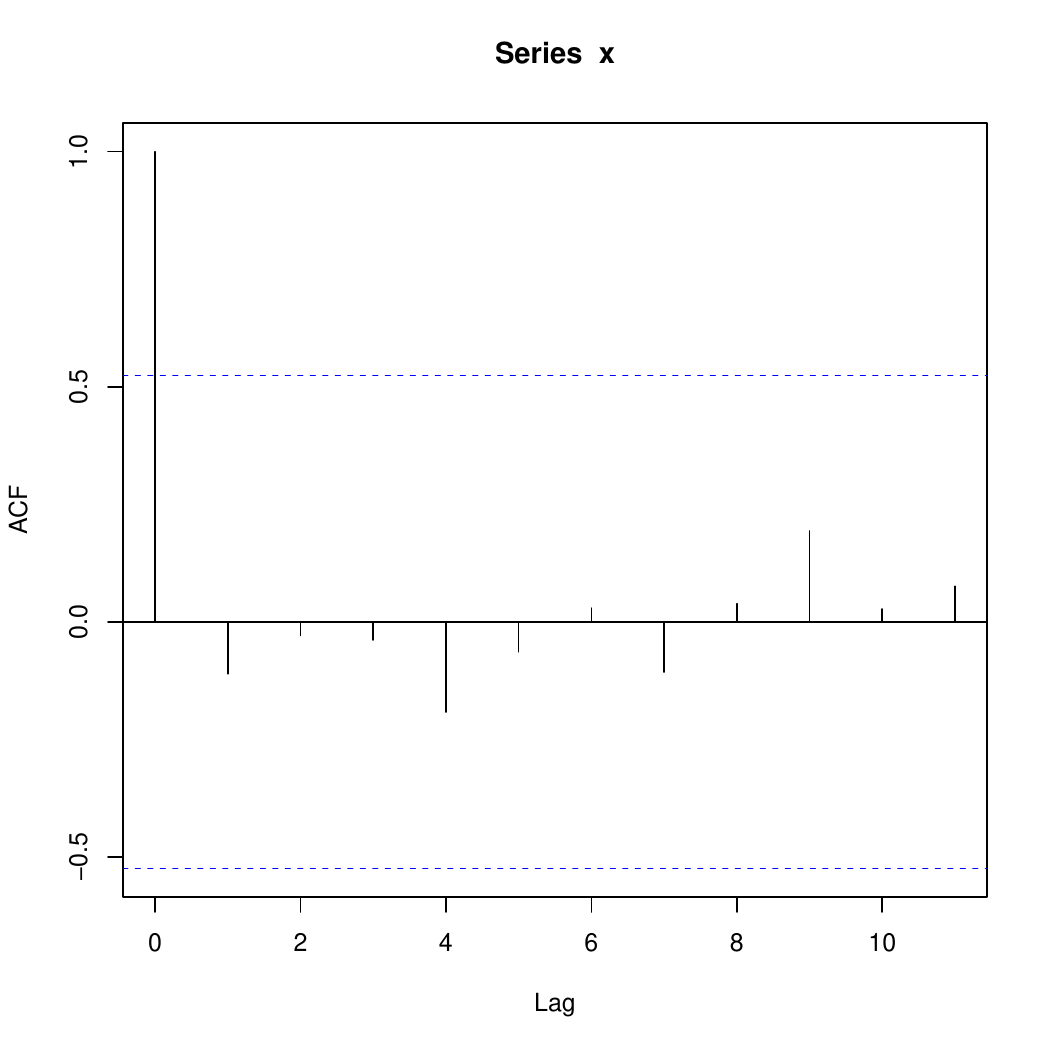}
		\label{fig:l3}}
	\hfill
	\subfloat[ACF of first order difference of monthly price per request.]{%
		\includegraphics[width=5.5cm]{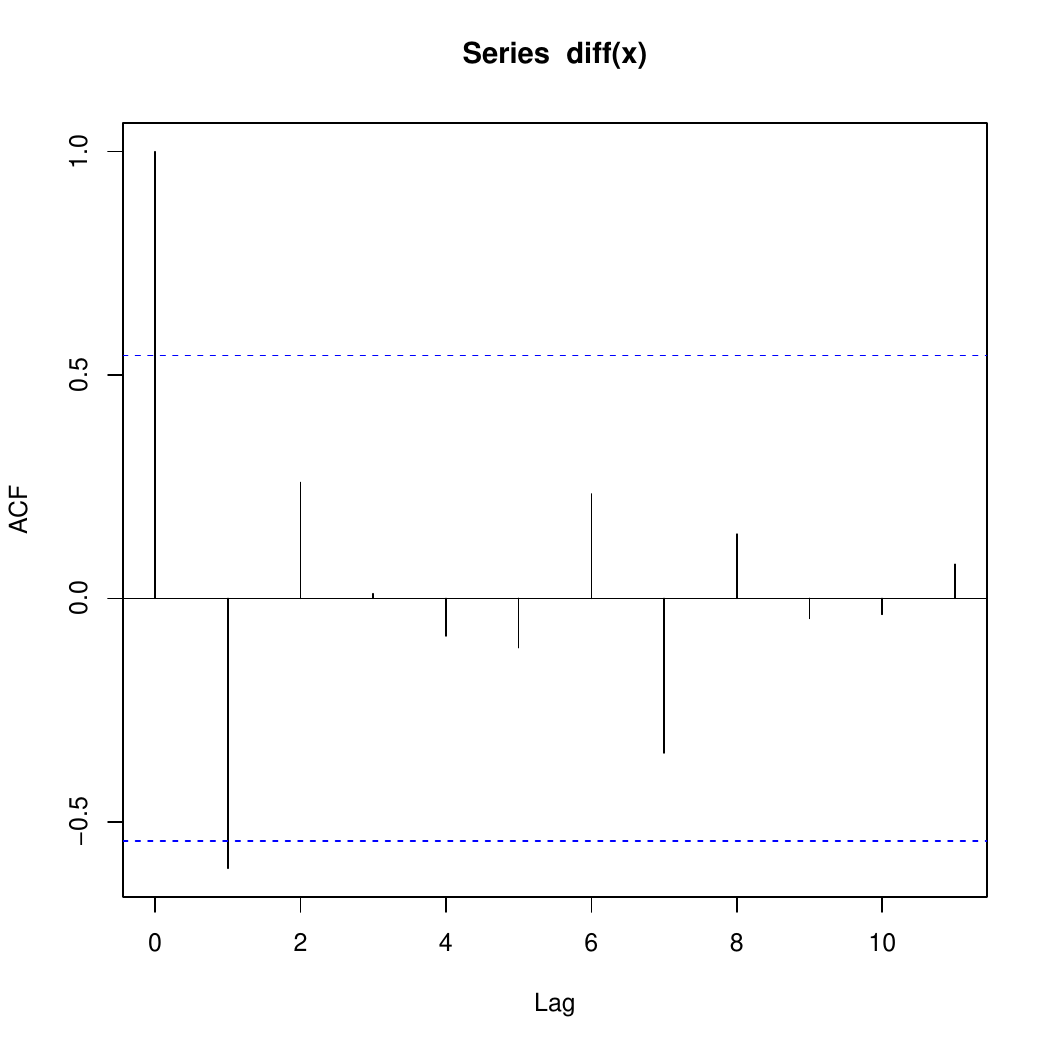}
		\label{fig:l4}}
	\caption{Autocorrelation functions (ACF) with time lags.}
	\label{acfplot}
\end{figure}
Before going any further into AR and MA models, we plot the autocorrelation function (acf) first. In the left panel of figure (\ref{acfplot}) we see the movement of the trend over time (in months). On the vertical axis we plot the acf defined in (\ref{1}) and on the horizontal axis we measure the time with lag, where lag=0 represents the time period t, and lag=2 then represents time $t-2$ (in which the vertical line corresponding to its autocorrelation is $\rho_{t,t-2}$). In short, the autocorrelation value with lag=$q$ can be obtained by the height of the vertical line, $\rho_{t,t-q}$. In panel \ref{fig:l3}, if we notice a trend of growing height of the vertical lines, we do not have convergence towards zero at a higher time lag. This implies that our dataset is not trend stationary. On the other hand, if we take the first order difference of monthly price per request (i.e., $y_t-y_{t-1}$), and perform a similar autocorrelation analysis, we would be able to see some convergence as shown in panel \ref{fig:l4}. This is our desirable case as in the behavior of people is myopic. In our case this myopic behavior implies recent monthly price per request of recent months has more impact than earlier months in explaining current monthly price per request. Moreover, figure \ref{fig:l4}, the heights of the vertical lines are \\
(i) very high for lag between 0 and 2,\\
(ii) have relatively shorter height for lag between 3 and 5,\\
(iii) have increased height between the lag of 6 and 8, and \\
(iv) fell thereafter.\\
This clearly exhibits seasonality in the first difference in monthly price per request data.

To test for stationarity of monthly price per request, we utilize the \emph{Augmented Dicky-Fuller} and \emph{Phillips-Perron} unit root tests. Consider the model
\begin{align}\label{2}
\Delta\ y_t &= \mu +\beta\ t + \gamma y_{t-1}+\delta_1\ \Delta y_{t-1}+ \delta_2\ \Delta y_{t-2}+\ ...\ + \delta_{q-1}\ \Delta y_{t-(q-1)}+\epsilon_t,
\end{align}
where $\Delta y_{t-i} = y_{t-i}-y_{t-(i-1)},\ \forall\ i=0,1,...,14$, $\gamma$ is the coefficient of price per request at time $(t-1)$, $\mu$ is the mean, $\beta$ is coefficient on the time trend and $\delta_i$ is the coefficient of the $i^{th}$ lag. Our main objective is to test if $\gamma$ is zero. If it takes the value zero then we conclude that, the system is not trend stationary. To check it we perform  \emph{Augmented Dicky-Fuller} unit root test. Here the null hypothesis is
\[
\mathscr{H}_0:\ \gamma=0,
\]
versus the alternate hypothesis
\[
\mathscr{H}_1:\ \gamma<0.
\]
In this test we assume under the null hypothesis (i.e., $\mathscr{H}_0$) that unit root is present and under the alternative hypothesis (i.e.,$\mathscr{H}_1$) we assume that monthly price per usage has trend stationarity. The \emph{Augmented Dicky-Fuller} test statistic is defined as 
\begin{equation}\label{3}
ADF_\gamma=\frac{\tilde \gamma}{S.E.(\tilde \gamma)},
\end{equation}
where $\tilde\gamma$ is the predicted value of $\gamma$ after running the ordinary least square method of equation (\ref{2}), where $S.E.(\tilde \gamma)$ is the standard error of the predicted value of $\gamma$. Under this test if at certain percentage level of significance (i.e., $\alpha$) the P-value$>\alpha$ we do not have enough evidence against the null hypothesis that $\mathscr{H}_0: \gamma=0$. Therefore, at that $\alpha$-level of significance monthly price per request has unit root problem and then monthly price per request is not stationary under that $\alpha$-level of significance.

On the other hand, \emph{Phillips-Perron} unit root test considers the following model,
\begin{equation}\label{4}
\Delta y_t=\gamma\ y_{t-1}+\epsilon_t.
\end{equation}
The \emph{Phillips-Perron} test shares the same null and alternative hypotheses as the \emph{Augmented Dickey-Fuller} test and also uses the same t-statistic.

The key difference between the two tests is that the \emph{Phillips-Perron} test is non-parametric, whereas the \emph{Augmented Dickey-Fuller} test accounts for serial correlation. As a result, the \emph{Phillips-Perron} test tends to perform better with large datasets. Therefore, it is recommended to apply both tests before proceeding further. In time series analysis, there is no single test statistic that can conclusively address the unit root problem in a system. Both the \emph{Augmented Dickey-Fuller} and \emph{Phillips-Perron} tests share drawbacks, such as sensitivity to structural breaks and weak power in small samples, which can lead to a conclusion of a unit root. Since both tests yielded the same result in this case, further testing is unnecessary. In both tests, we use the p-value to determine whether to reject the null hypothesis. We set the level of significance $\alpha$ at 0.05, a standard value, and test the hypothesis at a 95\% confidence level. In both the \emph{Augmented Dickey-Fuller} and \emph{Phillips-Perron} unit root tests, the p-value was 0.01688, which is less than 0.05. Therefore, at a 95\% significance level, we have enough evidence to reject the null hypothesis, indicating that the first difference of the monthly price per request has a unit root problem. Thus, the first difference is stable, as seen in panel \ref{fig:l4} of figure \ref{acfplot}. We opted not to perform a \emph{Dickey-Fuller} unit root test on the monthly price per request data because the sample size is less than 25, and the \emph{Dickey-Fuller} test does not perform well with such small samples.

\section{Model Selection.}

In this section, our objective is to determine the statistical modeling approach that most effectively characterizes the monthly price per request series. The analysis begins with an evaluation of widely used time series models, focusing on their ability to capture the underlying dynamics and temporal dependencies present in the data. Through this initial assessment, the autoregressive model of order one, AR(1), emerges as the most appropriate choice, offering the best fit among the standard models considered. However, the selection process does not conclude with model fit alone. To ensure that the chosen specification also delivers optimal predictive performance, we extend our search to include alternative statistical models that may yield a lower MSE than the AR(1) benchmark \citep{brockwell2005modified,yusuf2025predictive}. This secondary evaluation emphasizes prediction accuracy over in-sample fit, with the ultimate aim of identifying a modeling framework capable of producing the most stable and reliable median estimates for forecasting the monthly price per request. Such stability is essential, as it directly informs the accuracy and fairness of the request-based pricing structure developed later in the study.

\subsection{Time Series Analysis.}
To obtain a stable expected price per request, we analyze various time series models to determine the best fit for our dataset of monthly price per request spanning fourteen months. Rather than reviewing numerous models such as autoregression, moving averages, random walks, and exponential smoothing individually, we focus on the autoregressive integrated moving average model [ARIMA(p,d,q)] with different values of $p, d$, and $q$, where $p$ is the number of autoregressive terms, $d$ is the number of nonseasonal differences needed to be stationary, and $q$ is the number of lagged forecast errors in the prediction equation \citep{box1970distribution}. In terms of monthly price per request $y_t$, the general forecasting equation is
\begin{equation}\label{arima}
\hat y_t=\mu+\rho_1\ y_{t-1}+ ...+\rho_p\ y_{t-p}-\theta_1\ \epsilon_{t-1}-\ ...-\theta_q\ \epsilon_{t-q}.
\end{equation}
In equation (\ref{arima}) $\rho$'s and $\theta$'s are parameters of AR and MA, respectively. Let us first start with the autoregressive coefficients $\rho$'s. Suppose the value of $t$ is January, 2017. Then $y_{t-1}$ is monthly price per request in December, 2016. Now $\rho_1=\rho_{t,t-1}$ is the coefficient representing percentage change effect of $y_{t-1}$ on expected monthly price per request in January, 2017. From our previous discussion we know that this is nothing but the autocorrelation between time $t$ and $(t-1)$. These coefficients exist until $p^{th}$ time lag \citep{hua2019}. If we just stop here and add only the error term $\epsilon_t$ then this model becomes AR(p). Now, assume that the errors are also serially correlated. Thus, instead of $\epsilon_t$ we write $-\theta_1\ \epsilon_{t-1}-\ ...-\theta_q\ \epsilon_{t-q} $ with $\theta$'s represent how the error terms are related over time. We call $\epsilon_i$ as $i^{th}$ white noise. There are few cases normally examined when examining an ARIMA modeling problem, and we compute the AIC for each.  By the definition of AIC we know that if we have a statistical model M of monthly price per request data set, and if $m$ is the number of parameters and $\hat L$ the maximized likelihood value of that model (i.e., $\hat L=\ P(y_t|\hat\kappa, M)$) then $AIC=2\ m-2\ ln(\hat L)$. In this case $\hat\kappa$ is the vector of estimated parameter values of the model M. For the ARIMA model, the vector $\hat{\kappa}=(\mu\ \rho_1\ \rho_2\ \rho_3\ \dots\ \rho_p\ \theta_1\ \theta_2\ ...\ \theta_q)'$. In R file \emph{aics\_tests\_with\_smaller\_data.R} we did calculate the AIC values and the unit root tests.
\begin{table}[htbp]
	\caption{Performance of different autoregressive integrated moving average models}
	\centering
	\begin{tabular}{c c}
		\hline \hline
		ARIMA(p,d,q) & Akaike information criterion\\
		\hline
		First order autoregression= ARIMA(1,0,0) & -166.50\\
		First order moving average= ARIMA(0,0,1) & \textbf{-168.30}\\
		Random walk= ARIMA(0,1,0) & -148.86\\
		Differenced first-order autoregressive model= ARIMA(1,1,0) & -151.84\\
		Simple exponential smoothing= ARIMA(0,1,1) & -152.39\\
		Second order autoregression= ARIMA(2,0,0) & -164.74\\
		Second order moving average= ARIMA(0,0,2) & -167.23\\
		First order autoregressive moving average= ARIMA(1,0,1) & -166.88\\
		Autoregressive moving average (2,3)= ARIMA(2,0,3) & -166.57\\[1ex]
		\hline
	\end{tabular}
	\label{table:AIC}	
\end{table}

The model which has the smallest AIC is the best fit, and as such, we clearly see (see table \ref{table:AIC}) that first order moving average[ARIMA(0,0,1)] is our best fit model. Furthermore, we compute the median value of this model to be $0.00812937$, considering the median to mitigate the effect of outlying data. Another way to see a model is a good fit of a data is to check its MSE. By definition 
$$MSE = T^{-1}\sum_{t=1}^T\ (y_t-\hat y_t)^2,$$
where $y_t$ and $\hat y_t$ are the actual and the predicted values of monthly price per request at $t^{th}$ month, respectively. In the above equation $T^{th}$ rank represents the end time point of an experiment which in our case is December, 2016. As we are considering the monthly data, we are assuming the difference of each time period is same. Size of our data of experiment is 14, which makes $T=14$. In this model our MSE is $2.190918\ \times\ 10^{-7}$.

\subsection{Cubic Smoothing Spline Analysis.}
From the previous section we know that, by using the time series analysis MA(1) gives us the lowest AIC and the corresponding MSE is  $2.190918\ \times\ 10^{-7}$. Now we try to find if there exist a model which might give lower MSE. From statistical literature we know, lower MSE implies better fit of a model but, we have to remember that by concentrating too much on MSE we might end up with the problem of over fitting. Throughout our studies we consider unconstrained optimization (like, AR is the time series version of simple ordinary least square(OLS) method) through squared error minimization. Instead of just focusing on unconstrained optimization let us introduce some constraint of the function and multiply with a penalize parameter $\lambda$ (or the roughness parameter) and construct a Lagrangian constrained optimization problem. By doing this we will achieve a stable estimator of the coefficients with lower standard error.

Generally in regression analysis for a set of data ($x_i,Y_i$) $\forall i\in\mathbb{Z}$ and for all $x_1<x_2<\ ...\ <x_n$ we model 
\begin{equation}
\label{genmod1}
Y_i=g(x_i)+ \epsilon_i
\end{equation} 
Here, $\hat g(.)$ is a function of $x_i$ which is unknown, and $\epsilon_i$ is the error corresponding to the model in (\ref{genmod1}). Therefore, in equation (\ref{genmod1}) 
$$\epsilon_i=Y_i-\hat g(x_i).$$
In particular in our case we have just only one independent variable which is time. To find out the best possible fit of the above specified model, we  minimize the squared term of the error, to remove negative values of the differences. Hence, we  minimize 
$$\epsilon_i^2=\sum_{i=1}^T\ (y_{i}-\hat g(x_i))^2,$$ 
which is an unconstrained optimization. Following spline regression literature  we add a constraint $\int_{x_1}^{x_n}\ \hat g''(x)^2\ dx$, a roughness penalty, which can be thought of as the second order condition of the $\hat g(x)$. This kind of constraint is quite intuitive. $\hat g''(x)$ is the nothing but the second order derivative and which is the curvature of the function. If we minimize the error with the constraint of curvature of itself we can get some shrinkage of each coefficient corresponding to each month (as it is non-parametric) \citep{pramanik2024motivation}.

Let us discuss it in terms of simple linear regression. Then probably the concept of the spline regression would be more clear. Suppose, under the OLS $\hat g(x_i)=\beta\ x_i$, where $\beta$ is the coefficient corresponding to $x_i$ regressor. In this scenario $$\epsilon_i^2=(y_i-\beta\ x_i)^2.$$ By OLS our primary objective is to, $$\min_{\beta}\ \sum_{i=0}^n\ (y_i-\beta\ x_i)^2.$$ According to \citep{tibshirani1996} the objective function can be written as $$(\beta-\hat{\beta})'\ X'X\ (\beta-\hat{\beta})$$
which creates the \emph{elliptical} contours and doing the minimization of sum of squared error term stays at the lowest contour point. Now suppose we introduce a constraint $||\beta||^2$ first and $||\beta||$ in the second and see how the behavior of the minimization got change. In page 271 of \cite{tibshirani1996} shows that for the first constraint the minimization point touches the surface of the sphere to the upper contour than we got in the unconstrained minimization. For the second constraint we get the similar result but instead of sphere now the constraint looks like a square and the minimization occurs at the corner of that square. The objective function with $||\beta||^2$ is called the Ridge regression and the objective function with $||\beta||$ is called Least absolute shrinkage and selection operator (LASSO).

Since cubic smoothing spline is based on a non-parametric method we can not see the $\beta$'s directly. Hence,  the scenario is very similar to ridge regression. In this case,  the Lagrangian of the objective function becomes
\begin{equation}
\label{lagrange1}
\mathcal{L}=\sum_{i=1}^T\ (y_{i}-\hat g(x_i))^2\ +\lambda\ \int_{x_1}^{x_n}\ \hat g''(x)^2\ dx,
\end{equation}
where $\lambda$ is the smoothness parameter, $\lambda\in[0,\infty)$. If $\lambda\rightarrow 0$ then there is  no smoothing and we end up with an interpolating spline. On the other hand, if $\lambda\rightarrow\infty$ then the roughness parameter is huge and converges to linear least square estimates.
\begin{figure}[htbp]
	\centering
	\includegraphics[width=9cm]{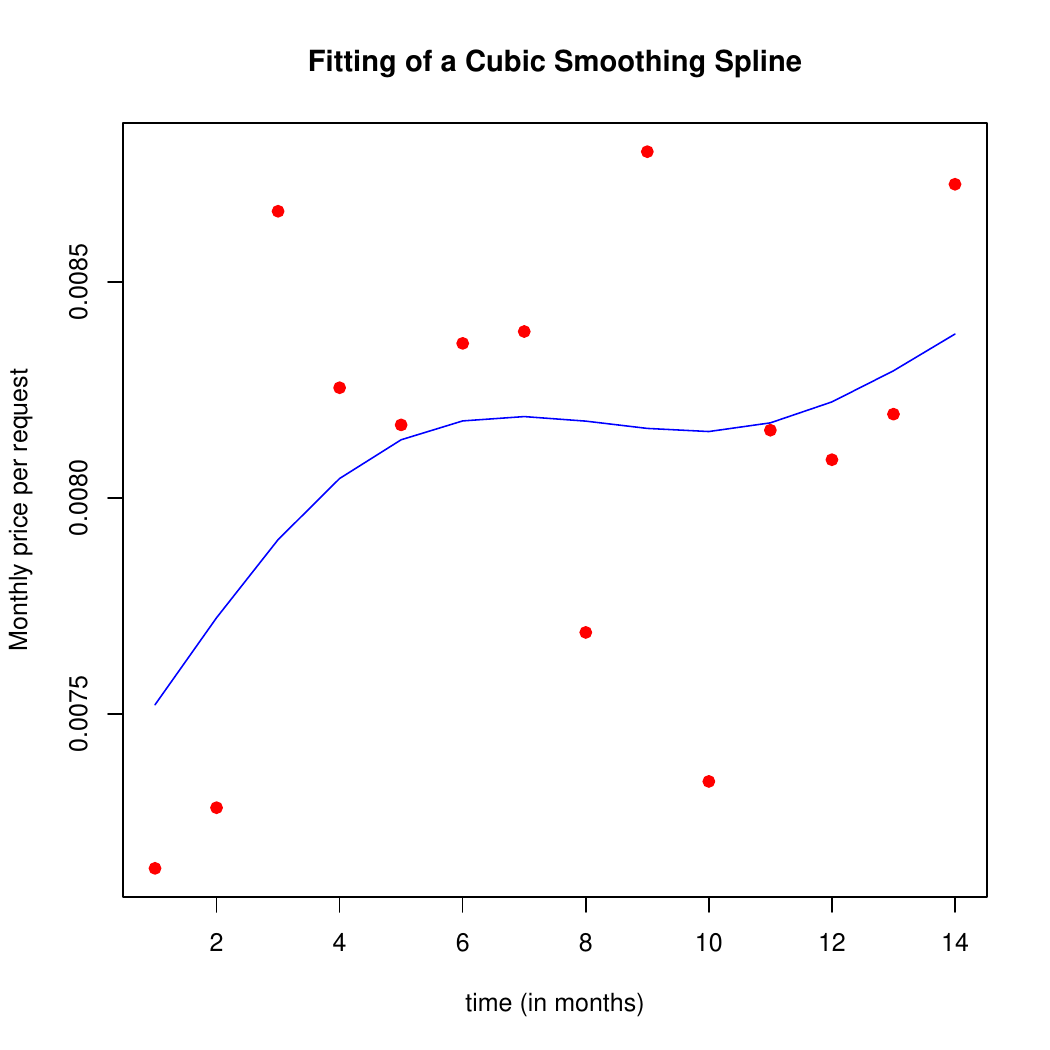}
	\caption{Spline fit for our model, where $\lambda$ follows Inverse gamma distribution with hyperparameters $\alpha=10$ and $\beta=4.7988$}
	\label{spline}
\end{figure}
In figure \ref{spline} we fit a cubic smoothing spline with our data set of monthly price per request. Horizontal and vertical axes represent time (in months) after October 2015 and monthly price per request, respectively. This model gives the MSE as $1.774695\times \ 10^{-7}$ which is significantly lower than we had for MA(1). To reduce the outlier effect, we take the median of the predicted value obtained from the model, $0.008168$, which is the predicted stable monthly price per request, to be used in the tier pricing model in the next sub-section.

\subsection{The Tier Pricing Model.}
While transitioning to a usage-based pricing model may be significantly fairer, a system with too many components could become difficult to explain and for customers to comprehend. Therefore, we have established the following guidelines for our pricing model:\\
(a). Any usage-based system should avoid invoices below a minimum threshold.\\
(b). To build a model that is both easy to explain/understand and keeps overall invoicing the same, our model will take a tiered structure.\\
(c). Since too many tiers could lead to confusion, we restrict to no more than 10 tiers.\\
Reviewing the usage of ProblemIT by eCW clients across every month of 2016, we split clients into 8 tiers (see table \ref{table:tierprice}) based on search request percentiles:
\begin{table}[htbp]
	\caption{Proposed tier pricing model}
	\centering
	\begin{tabular}{c c c c}
		\hline \hline
		Tier & Percentile & Number of search request & Price (per month) \\
		\hline
		1 & 0-25\% & 0-535 & \$4.37\\
		2 & 25\%-35\% & $<$ 855 & \$ 6.98 \\
		3 & 35\%-45\% & $<$ 1280 & \$ 10.45\\
		4 & 45\%-55\% &  $<$ 1870 & \$ 15.27\\
		5 & 55\%-65\% & $<$ 2765 & \$ 22.58\\
		6 & 65\%-78\% & $<$ 4620 & \$ 37.74\\
		7 & 78\%-85\% & $<$ 6975 & \$ 56.97\\
		8 & 85\%-100\% & $>$ 6975 & \$ 0.008168$\times$ number of requests\\[1ex]
		\hline
	\end{tabular}
	\label{table:tierprice}	
\end{table}
We set the price in each tier by multiplying the high end of the search usage in each tier by the price per request that was calculated using the semi-parametric Bayesian cubic smoothing spline from Section 5.1 (see figure \ref{spline}).

\section{Significance.}
Table \ref{table:invoice} represents the the actual monthly invoices and new suggested invoices for first four months of 2017 \footnote{eClinicalWorks. \url{https://www.eclinicalworks.com/products-services/ehr/}}. By table \ref{table:invoice} we know, if the proposed pricing had been implemented in January 2017, the estimated income for eCW ProblemIT would have been \$239,738.07, compared to the actual income of \$252,473 for that month \cite{evans2016electronic}. To illustrate this, Table \ref{table:invoice} compares the invoices received at the beginning of 2017 with what would have been generated under the suggested tiered pricing model. The decline in income is attributed to the seasonal nature of ProblemIT's usage; if the price per request remains constant, lower usage results in reduced income. However, it is possible that an increase in usage in the later months of 2017 could offset these losses.
\begin{table}[htbp]
	\caption{Actual invoice versus projected invoice}
	\centering
	\begin{tabular}{c c c c c}
		\hline \hline
		Month & Actual Invoices & New Pricing Invoices & Difference & Percentage Change\\
		\hline
		January & 252473.10 & 239738.073 & 12735.010 & -5.04\\
		February & 251774.00 & 223626.410 & 28147.600 & -11.18\\
		March & 253151.80 & 246359.841 & 6791.974 & -2.68\\
		April & 252576.50 & 225071.387 & 27505.120 & -10.89\\[1ex]
		\hline
	\end{tabular}
	\label{table:invoice}	
\end{table}
The number of customers who will experience an increase in their monthly payments is nearly equal to those who will pay less: 2,333 will pay more, while 2,312 will pay less \footnote{eClinicalWorks Pricing. \url{https://www.eclinicalworks.com/products-services/pricing/}}. For customers facing an increase, the average rise is \$57.59, which translates to about \$10.70 per reported user each month. In contrast, customers who will pay less with this pricing will save an average of \$41.57, equating to approximately \$1.96 per reported user per month. It is evident that clients who frequently under-report have a significant increase in their prices, which raises the average change, while most clients see only minor adjustments per reported user. A detailed breakdown is provided below.
\begin{figure}[htpb]
	\centering
	\centering
	\includegraphics[width=13cm]{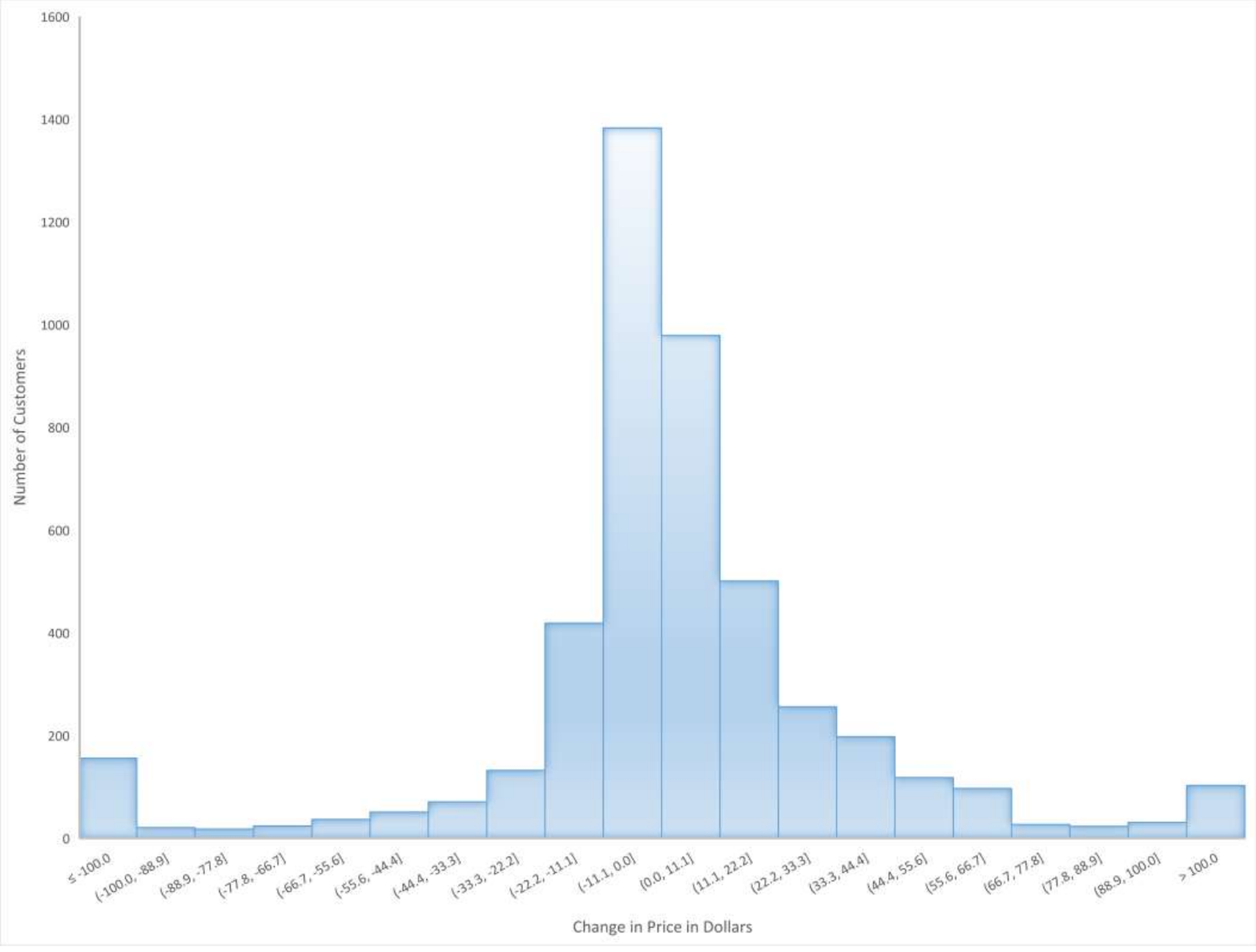}
	\caption{Dollar change in price in January, 2017}
	\label{money}
\end{figure}

\medskip

The tallest bar in figure \ref{money} represents the number of customers that will have their price per month decrease between zero dollars and \$11.10. To the left are all of the groups that will have their price decrease and the groups to the right will have their price per month increase. At the extreme ends of the graph we have overflow buckets of price increases of more than \$100 on the right side and decreases of more than \$100 on the left side.
\begin{figure}[htbp]
	\centering
	\includegraphics[width=14cm]{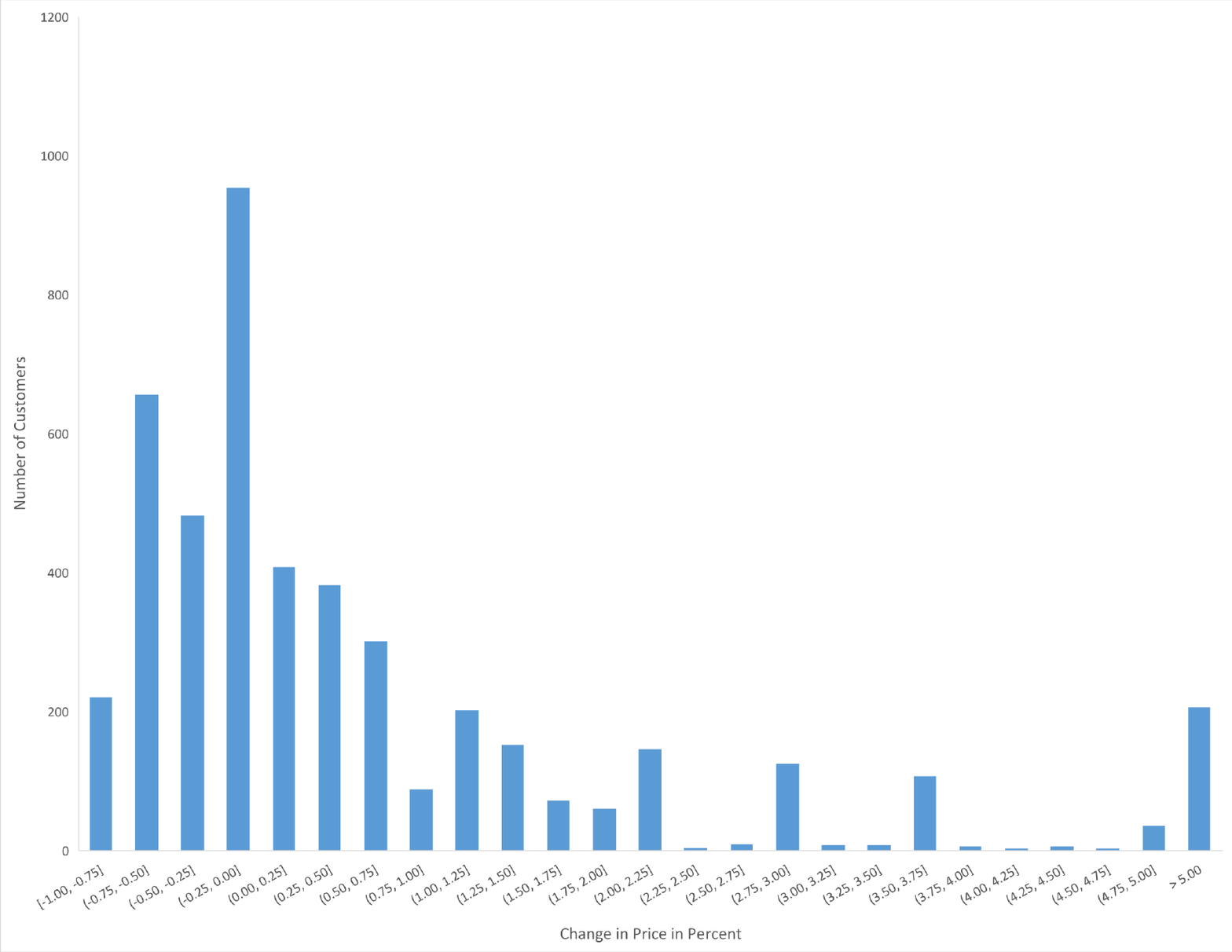}
	\caption{Percentage change in price}
	\label{price}
\end{figure}
Figure \ref{price} shows the same customers as in figure \ref{money}, but illustrates the changes in price as percentages. The graph is inherently skewed to the right because prices cannot decrease by 100\%, while they can increase by more than 100\%. The first four groups on the left will see a decrease in their prices, whereas the remaining groups consist of customers who will experience an increase \footnote{eClinicalWorks Pricing. \url{https://www.eclinicalworks.com/products-services/pricing/}}.

The customers that will have the highest increase in price per user have previously been identified as possible serial under-reporters in previous analysis done by BI when analyzing eCW provider-counts. Figures \ref{orlando} and \ref{sovereign} are based on \cite{ecl}.  For example, two notable examples are Orlando Family Physicians (GUID: e0986b8766d849f5)(see figure \ref{orlando}) and Sovereign Health Medical Group (GUID: 293675d59ef717af)(see figure \ref{sovereign}). Orlando Family Physicians reports having 30 end users, while Sovereign Health Medical Group reports only having one user. We will look at their 2016 usage in figure \ref{orlando}; it can be noted the top most dotted line represents $6x$ the ``typical'' usage for a customer with the same number of licenses as these example customers.

\begin{figure}[htbp]
	\centering
	\centering
	\includegraphics[width=12cm]{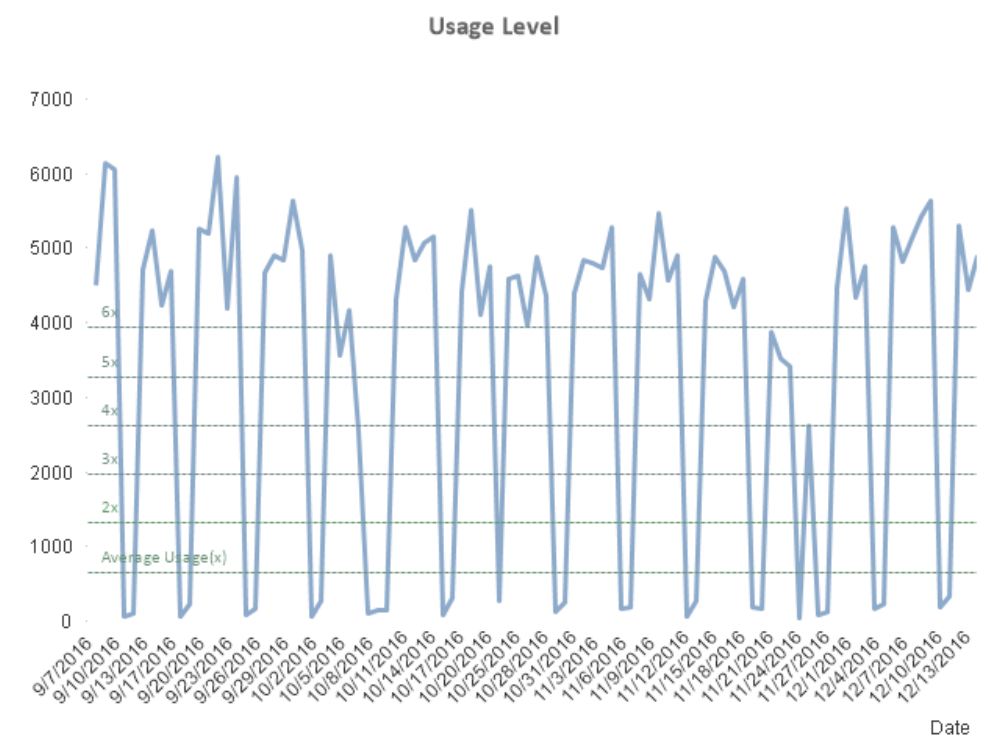}
	\caption{Orlando Family Physicians}
	\label{orlando}
\end{figure}

\begin{figure}[htbp]
	\centering
	\centering
	\includegraphics[width=12cm]{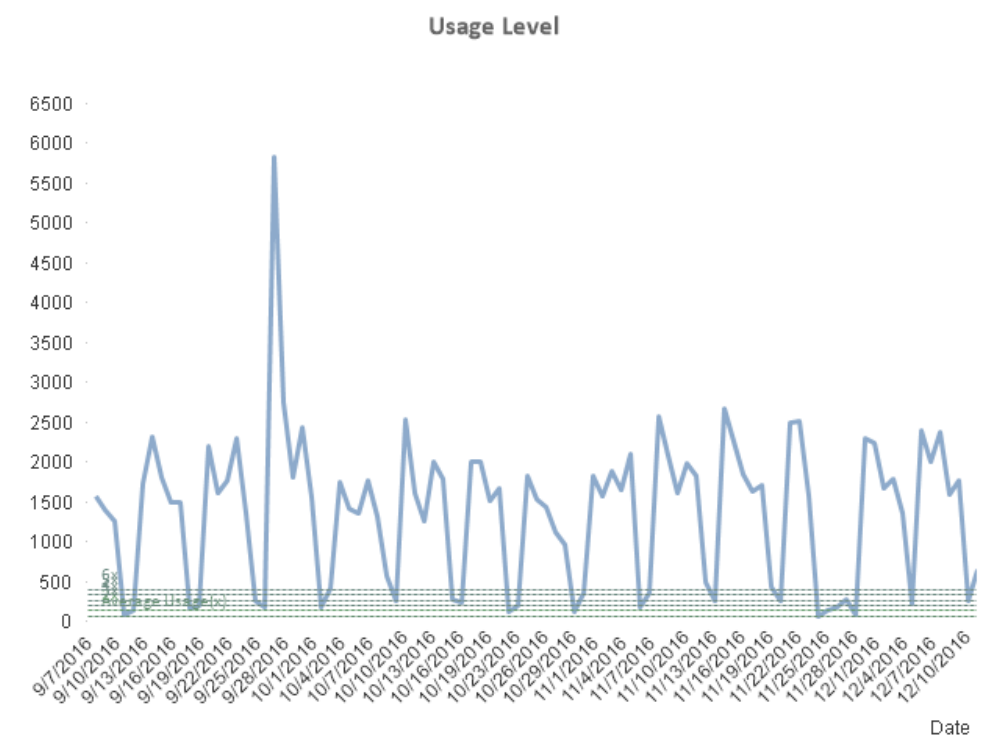}
	\caption{Sovereign Health Medical Group}
	\label{sovereign}
\end{figure}
It is clear that both customers have usage levels significantly higher than average, given the number of licenses they are paying for. In January 2017, Orlando Family Physicians paid \$162.60, but under the tier pricing model, they would have incurred a charge of \$1,638.40 (see figure \ref{orlando}). Similarly, Sovereign Health Medical Group paid \$5.42 in January 2017, whereas under the tier pricing model, their payment would have been \$368.42 (see figure \ref{sovereign}).

\section{Concluding Remarks.}

An examination of the projected cost implications under the proposed request-based pricing structure reveals that a significant number of accounts would experience an increase in their total charges. Specifically, 806 accounts report having only a single licensed user, while an additional 470 accounts report having two licensed users. This finding underscores a fundamental imbalance within the current user-based pricing framework, as it suggests that many organizations pay for licenses covering only one or two individuals despite usage patterns indicating that considerably more people are actively interacting with IMO’s Problem(IT) portal. Such discrepancies point toward the likelihood of under-reported user counts or shared account usage, which allows organizations to access the system far beyond the scale implied by their reported license numbers \citep{pramanik2025optimal}. In a purely license-based model, these clients may appear to be low-volume users in financial terms, yet their actual consumption of system resources places a disproportionately high demand on infrastructure and support services. The proposed request-based model seeks to address this misalignment by ensuring that costs are tied directly to the scale of system utilization rather than solely to the nominal number of licensed accounts \citep{pramanik2025construction}. By aligning charges with actual usage, the model aims to produce a more equitable cost distribution across the client base, improving both fairness and revenue recovery.

However, while the tiered, request-based approach addresses important issues of fairness and resource alignment, it also introduces certain potential challenges that must be considered carefully in the context of long-term customer engagement and product adoption. One possible drawback, supported by prior research \citep{ramsay1991some}, is that charging on a per-request basis may lead to a reduction in overall usage, as customers become more cautious in their system interactions to avoid incurring higher charges. From a financial standpoint, this could reduce transaction volume and limit revenue growth potential; more importantly, from a strategic perspective, it may hinder the broader organizational objective of promoting consistent and extensive use of all IMO products. This concern is particularly relevant for accounts that would see substantial cost increases under the new tiered framework. While such customers may currently be significantly under-reporting their actual user counts, the imposition of higher charges could still generate resistance, with some clients potentially reducing their reliance on the platform or even seeking alternative solutions \citep{pramanik2024measuring}. Consequently, we contend that the overarching goal of any revised pricing strategy should not be to discourage utilization, but rather to create a balanced structure that incentivizes continued engagement while ensuring that charges reflect actual system consumption. This will require careful calibration of the tier thresholds and per-request pricing levels, as well as proactive customer communication to foster understanding and acceptance of the revised model.

To address these potential drawbacks and ensure that the transition to a request-based pricing framework is both fair and sustainable, several mitigation strategies should be considered. One option is to adopt a hybrid pricing structure that combines a modest base fee per licensed user with a variable component tied to actual request volumes. This would allow customers to maintain predictable baseline costs while still accounting for high-volume usage in a proportionate manner. Another approach is to implement volume-based discounts or usage credits for clients whose request activity surpasses certain thresholds, thereby reducing the risk of sudden and substantial cost increases while still incentivizing efficient system use. Such measures could be particularly valuable for organizations that make extensive use of IMO’s Problem(IT) portal for legitimate operational needs, ensuring they remain engaged without perceiving the pricing change as punitive. Additionally, a phased rollout of the revised model could help smooth the transition by allowing customers time to adapt to the new structure, assess their usage patterns, and adjust workflows accordingly. Clear, proactive communication of the rationale behind the new model emphasizing fairness, transparency, and alignment between costs and resource utilization will also be essential for building customer trust and acceptance. By incorporating these safeguards, the request-based model can achieve its dual objectives: creating a more equitable distribution of costs while fostering long-term customer satisfaction and sustained engagement with IMO’s suite of products.

The introduction of a revised pricing framework that carries the potential to reduce engagement with IMO’s Problem(IT) platform and, in turn, decrease the company’s revenue, presents a series of strategic and operational challenges that warrant careful evaluation before implementation. A key concern is the risk that higher charges could inadvertently deter usage among certain customer segments, particularly those already under-reporting their actual level of system interaction \citep{valdez2025association}. To address this, it may be prudent to initially modify the pricing structure in a targeted manner, applying adjustments only to accounts where there is compelling evidence of substantial under-reporting. Such a selective approach would allow for a more measured transition, reducing the likelihood of alienating lower usage or accurately reported accounts while still addressing the fairness and revenue recovery issues at the core of the model. Over time, the company could periodically reassess whether a universal, fully usage-based framework remains the most appropriate long-term strategy, informed by data on customer behavior, retention, and revenue stability. However, the effectiveness of this selective adjustment hinges on the ability to accurately monitor user counts, a process that if handled manually could impose considerable demands on internal staff and resources \citep{valdez2025exploring}. To mitigate these burdens, an automated monthly auditing system could be deployed, enabling ongoing surveillance of license utilization and request activity without significantly increasing operational overhead. Such a system would allow for the continuous identification of accounts with substantial discrepancies between reported license counts and observed activity, thereby supporting the fair application of the revised model while preserving organizational efficiency.

The introduction of a revised pricing framework that carries the potential to reduce engagement with IMO’s Problem(IT) platform and, in turn, decrease the company’s revenue, presents a series of strategic and operational challenges that warrant careful evaluation before implementation. A key concern is the risk that higher charges could inadvertently deter usage among certain customer segments, particularly those already under-reporting their actual level of system interaction. To address this, it may be prudent to initially modify the pricing structure in a targeted manner, applying adjustments only to accounts where there is compelling evidence of substantial under-reporting. Such a selective approach would allow for a more measured transition, reducing the likelihood of alienating lower-usage or accurately reported accounts while still addressing the fairness and revenue recovery issues at the core of the model \citep{khan2023myb0}. Over time, the company could periodically reassess whether a universal, fully usage-based framework remains the most appropriate long-term strategy, informed by data on customer behavior, retention, and revenue stability. However, the effectiveness of this selective adjustment hinges on the ability to accurately monitor user counts, a process that---if handled manually---could impose considerable demands on internal staff and resources. To mitigate these burdens, an automated monthly auditing system could be deployed, enabling ongoing surveillance of license utilization and request activity without significantly increasing operational overhead. Such a system would allow for the continuous identification of accounts with substantial discrepancies between reported license counts and observed activity, thereby supporting the fair application of the revised model while preserving organizational efficiency.

The implementation of an automated monthly audit could be seamlessly integrated into the existing infrastructure of the IMO Data Warehouse, leveraging the same data sources already used in the development of the request-based pricing model. Key operational metrics such as license counts from the \texttt{channel\_sales\_fact} table, customer identifiers from the \texttt{customer\_dim} table, and portal request volumes from the \texttt{log.org\_stats} table---could be programmatically linked to produce a recurring compliance report. This report would identify customers whose actual request activity exceeds expected thresholds relative to their reported license counts, flagging cases of potential under-reporting for further review. Automation of this process would not only ensure consistent and timely monitoring but would also significantly reduce the manual effort otherwise required for such analysis. Additionally, by storing historical audit outputs, the system could track patterns of repeated discrepancies over time, enabling data-driven decisions regarding targeted adjustments to pricing structures or the initiation of customer outreach. Such a data-driven auditing mechanism would thus reinforce the transparency, fairness, and sustainability of the proposed pricing framework while aligning closely with the company’s operational capabilities.

\subsection*{Availability of data.}
Data sets  were obtained from Intelligent Medical Object's own database.
\subsection*{Competing interests.}
No potential conflict of interest was reported by the authors.	
\subsection*{Funding.}
No funding has been used to write this paper.

\section{Appendix.}

\subsection{Definitions and Theorems.}

\begin{thm}
	If $$y_t=\rho_{t,t-1}\ y_{t-1}+\epsilon_t$$ is an AR(1), and $$y_t=\epsilon_t+\theta\ \epsilon_{t-1}$$ is MA(1) , with $|\theta|<1$, then MA(1)=AR($\infty$) and vice versa.
\end{thm}
\begin{proof}
	\textbf{Case I: MA(1)=AR($\infty$).}\\

	Let us start with the first order moving average MA(1) with $$y_t=\epsilon_t+\theta\ \epsilon_{t-1}.$$ Here $\epsilon_t$ is the white noise error term of time t.
	Let define a lag or backshift operator L such that, $Ly_t=y_{t-1}\ \forall\ t>1$, or equivalently, $y_t=Ly_{t+1}\ \forall\ t>1$. The lag operator can be raised to arbitrary integer powers so that,
	$L^{-1}\ y_t=y_{t+1}$ and $L^\upsilon\ y_t=y_{t-\upsilon}$.
	Now,
	from the equation of the first order moving average we know that, 
	\begin{align}\label{a1}
	y_t&=\epsilon_t+\theta\ \epsilon_{t-1}\notag\\
	&=\epsilon_t-\theta\ L\ \epsilon_{t}\ [by\ using\ the\ definition\ of\ lag\ operator]\notag\\
	&=(1-\theta\ L)\epsilon_t\notag\\\implies\ \epsilon_t&=\frac{y_t}{1-\theta\ L},\ where\ |\theta|<1
	\end{align}
	
	From the definition of \emph{Cauchy} series we know for any variable $y$, $S_\infty=y+\theta\ y+ \theta^2\ y+...=y(1-\theta)^{-1}$, where $|\theta|<1$.
	After rewriting equation(\ref{a1}) we get,
	\begin{align}\label{a2}
	\epsilon_t&=\frac{y_t}{1-\theta\ L},\ where\ |\theta|<1\notag\\&=y_t+\theta\ L\ y_t+\theta^2\ L^2\ y_t+\theta^3\ L^3\ y_t+...\ to\ \infty\notag\\&=y_t+\theta\ y_{t-1}+\theta^2\ y_{t-2}+\theta^3\ y_{t-3}+...\ to \ \infty,\notag\\y_t&=y_t+\theta\ y_{t-1}+\theta^2\ y_{t-2}+\theta^3\ y_{t-3}+...+\epsilon_t
	\end{align}
	Equation (\ref{a2}) is nothing but AR($\infty$). Therefore, we get our desirable result.

	\textbf{Case II: AR(1)=MA($\infty$).}\\
	
	For $|\rho_{t,t-1}|<1$ we have our AR(1) model as, 
	\begin{align}\label{a3}
	y_t&=\rho_{t,t-1}\ y_{t-1}+\epsilon_t\notag\\y_{t-1}&=\rho_{t-1,t-2}\ y_{t-2}+\epsilon_{t-1}\notag\\y_{t-2}&=\rho_{t-2,t-3}\ y_{t-3}+\epsilon_{t-2},\ so \ on
	\end{align}
	Let us assume $\rho_{t,t-1}=\rho{t-1,t-2}=\rho{t-2,t-3}=\rho{t-3,t-4}=...=\rho$.
	After putting second equation in the place of $y_{t-1}$ of the first equation of the system (\ref{a3}) and increase up to power $\upsilon$ lag we get,
	\begin{align}\label{a4}
	y_t&=\rho_{t,t-1}\ y_{t-1}+\epsilon_t\notag\\&=\rho_{t,t-1}(\rho_{t-1,t-2}\ y_{t-2}+\epsilon_{t-1})+\epsilon_t\notag\\&=\rho^2\ y_{t-2}+\rho\ \epsilon_{t-1}+\epsilon_t,\ [as\ all\ \rho's\ are \ same]\notag\\&.\notag\\&.\notag\\&.\notag\\y_t&=\rho^\upsilon\ y_{t-\upsilon}+\sum_{k=0}^\upsilon\ \rho^k\ \epsilon_{t-k}, \ for\ |\rho|<1
	\end{align}
	As $|\rho|<1$ then $\lim_{\upsilon\rightarrow\infty}\ \rho^\upsilon\ y_{t-\upsilon}=0$. Therefore the first term of the right hand side of equation (\ref{a4}) vanishes and it becomes $y_t=\sum_{k=0}^\upsilon\ \rho^k\ \epsilon_{t-k}$ which is nothing but MA($\infty$).
\end{proof}

\subsection{Cubic Smoothing Spline.}
Generally in regression analysis for a set of data ($x_i,Y_i$) $\forall i\in\mathbb{Z}$ and for all $x_1<x_2<\ ...\ <x_n$ we model 
\begin{equation}
\label{genmod}
Y_i=g(x_i)+ \epsilon_i
\end{equation}
Here, $\hat g(.)$ is a function of $x_i$ which is unknown, and $\epsilon_i$ is the error corresponding to the model in (\ref{genmod}). Therefore, in equation (\ref{genmod}) 
$$\epsilon_i=Y_i-\hat g(x_i)$$
and to find a best fit model, we want to minimize the squared term of the error, to remove negative values of the differences. Hence, we wish to minimize 
$$\epsilon_i^2=\sum_{i=1}^T\ (y_{i}-\hat g(x_i))^2.$$ 
This is an unconstrained optimization. Under spline literature with this objective we add another constraint $\int_{x_1}^{x_n}\ \hat g''(x)^2\ dx$, a roughness penalty, which can be thought of as the second order condition of the $\hat g(x)$. Therefore, the Lagrangian of the objective function becomes
\begin{equation}
\label{lagrange}
\mathcal{L}=\sum_{i=1}^T\ (y_{i}-\hat g(x_i))^2\ +\lambda\ \int_{x_1}^{x_n}\ \hat g''(x)^2\ dx,
\end{equation}
where $\lambda$ is the smoothness parameter, $\lambda\in[0,\infty)$. If $\lambda\rightarrow 0$ then there is  no smoothing and we end up with an interpolating spline. On the other hand, if $\lambda\rightarrow\infty$ then the roughness parameter is huge and converges to linear least square estimates.

\begin{definition}
	A function $\hat g(x)$ is a spline of degree k on $[b,d]$ if
	\begin{itemize}
		\item $\hat g \in C^{k-1}[b,d]$ (continuous and $k$-times differentiable on our closed interval)
		\item $b=l_1<l_2<l_3< ... <l_n=d$ and 
		\begin{equation}
		\hat g(x)=\begin{cases}
		\hat g_0(x), & \text{$b=l_1\leq\ x\leq\ l_2$}\\ \hat g_1(x), & \text{$l_2\leq\ x\leq\ l_3$}\\ ...\\\hat g_{n-1}(x), & \text{$l_{n-1}\leq\ x\leq\ l_n=d$}
		\end{cases}\notag
		\end{equation}
	\end{itemize}
	where $ \hat g_{n-1}\in\mathbb{P}^k$, a $k^{th}$ degree polynomial.
\end{definition}

Taking a more specific version of this definition 
\begin{definition}
	A cubic smoothing spline $\hat g(x)$ is defined as 
	\begin{equation}
	\hat g(x)=\begin{cases}
	\hat g_0(x)=\alpha_0 x^3+\beta_0x^2+\gamma_0 x+\delta_0, & \text{$b=l_1\leq\ x\leq\ l_2$}\\
	\hat g_1(x)=\alpha_1 x^3+\beta_1x^2+\gamma_1 x+\delta_1, & \text{$l_2\leq\ x\leq\ l_3$}\\ ...\\\hat g_{n-1}(x)=\alpha_{n-1} x^3+\beta_{n-1}x^2+\gamma_{n-1} x+\delta_{n-1}, & \text{$l_{n-1}\leq\ x\leq\ l_n=d$}
	\end{cases}\notag
	\end{equation} 
	which satisfies
	
	$\hat g(x)\in C^2[l_1,l_n]$, and for all  i=1,2,...,(n-1) 
	\begin{equation}
	\ \hat g(x)=\begin{cases}
	\hat g_{i-1}(x_i)=\hat g_i(x_i)\\\hat g_{i-1}'(x_i)=\hat g_i'(x_i)\\\hat g_{i-1}''(x_i)=\hat g_i''(x_i)\\
	\end{cases}\notag
	\end{equation}
	where $\alpha_0,..,\alpha_{n-1},\beta_0,...,\beta_{n-1},\gamma_0,...,\gamma_{n-1},\delta_0,...,\delta_{n-1}$ are coefficient of the third degree polynomials.		
\end{definition}

The main reason behind using \emph{cubic smoothing spline} is that among all functions $\hat g \in C^2[b,d]$ interpolating ($Y_i,x_i$), it gives the smoothest curve, where the smoothness is measured by $\int_{x_1}^{x_n}\ \hat g''(x)^2\ dx$. Further more, as $\lambda\geq 0$ we assume it follows inverse gamma distribution with the probability density function (pdf)
\begin{equation}\label{ig}
f(x)=\frac{\beta^\alpha}{\Gamma(\alpha)}\ x^{-\alpha-1}\ \exp\left\{-\frac{\beta}{x}\right\}
\end{equation}
where $\beta=4.7988$ and $\alpha=10$. Throughout our analysis $x_i$ is nothing but different time periods (in months) and $y_i$ is monthly price per request. Now, after taking the mean of inverse gamma distribution as $\beta\ (\alpha-1)^{-1}$ (i.e., $0.5332$). We calculate the total number of knots by dividing this fourteen month into eight quantiles and take the length of it. As we have only two parameters $\lambda$ and total number of knots, assuming their distribution makes this spline as Bayesian cubic smoothing spline. Inverse gamma and nonparametric distributions are the prior distributions of this model. Furthermore, as We assume one parameter $\lambda$ has a parametric distribution and another parameter (i.e., total number of knots) has nonparametric distribution. This makes this model semi parametric Bayesian cubic smoothing spline.

\bibliographystyle{apalike}
\bibliography{bib}
\end{document}